\documentclass{article}
\usepackage{spconf,amsmath,graphicx}
\usepackage{amssymb}
\usepackage{amsfonts}
\usepackage{algorithm}
\usepackage{amsthm}
\usepackage{cite}
\usepackage{url}
\usepackage{epsfig}
\usepackage{subfigure}
\usepackage{psfrag}
\usepackage{color}

\title{An Efficient Algorithm for Device Detection and Channel Estimation in Asynchronous IoT Systems}
\twoauthors{Liang Liu}{EIE Department \\ The Hong Kong Polytechnic University \\ Email: liang-eie.liu@polyu.edu.hk}{Ya-Feng Liu}{LSEC, ICMSEC, AMSS \\ Chinese Academy of Sciences \\ Email: yafliu@lsec.cc.ac.cn}
\begin{document}
\ninept
\maketitle
\begin{abstract}
A great amount of endeavour has recently been devoted to the joint device activity detection and channel estimation problem in massive machine-type communications. This paper targets at two practical issues along this line that have not been addressed before: asynchronous transmission from uncoordinated users and efficient algorithms for real-time implementation in systems with a massive number of devices. Specifically, this paper considers a practical system where the preamble sent by each active device is delayed by some unknown number of symbols due to the lack of coordination. We manage to cast the problem of detecting the active devices and estimating their delay and channels into a group LASSO problem. Then, a block coordinate descent algorithm is proposed to solve this problem globally, where the closed-form solution is available when updating each block of variables with the other blocks of variables being fixed, thanks to the special structure of our interested problem. Our analysis shows that the overall complexity of the proposed algorithm is low, making it suitable for real-time application.

\end{abstract}

\newtheorem{example}{Example}
\newtheorem{corollary}{Corollary}
\newtheorem{definition}{Definition}
\newtheorem{lemma}{Lemma}
\newtheorem{theorem}{Theorem}
\newtheorem{proposition}{Proposition}
\newtheorem{remark}{Remark}
\newcommand{\mv}[1]{\mbox{\boldmath{$ #1 $}}}

\section{Introduction}\label{sec:Introduction}
Driven by the rapid advance of Internet of Things (IoT), massive machine-type communications (mMTC), the purpose of which is to provide wireless connectivity to a vast number of IoT devices, has attracted more and more attention recently. To reduce the device access delay, a grant-free random access scheme was advocated in \cite{Liu18}, where each active device first transmits its preamble to the base station (BS) and then directly transmits its data without waiting for the grant from the BS. To enable this low-latency access scheme, the BS should be able to detect the active devices and estimate their channels based on the received preambles \cite{Liu18,tutorial}. Recently, \cite{Part1,Chen18} show that joint device activity detection and channel estimation can be formulated as a compressed sensing problem because of the sparse device activity. Such a problem is then solved via the approximate message passing (AMP) algorithm \cite{donoho2009message} in \cite{Part1,Chen18,senel2018grant} or other sparse optimization techniques in \cite{jiang2018joint,ke2020compressive,ding2019sparsity,sun2019exploiting}. To practically implement joint device activity detection and channel estimation under the grant-free random access scheme, there are two crucial issues to address. First, in a practical mMTC system consisting of a large number of low-cost IoT devices, it is impossible to ensure that all the active devices are perfectly synchronized. Thereby, the preamble sequence sent by each active device may be delayed by some unknown number of symbols at the beginning of each coherence block. In this case, it is unknown that whether the device detection and channel estimation problem can be solved under the compressed sensing framework as for the synchronous case \cite{Part1,Chen18,senel2018grant}. Apart from the issue of asynchronous transmission, another challenge lies in the complexity. In mMTC, the number of devices is huge. Moreover, due to the recent success of the massive multiple-input multiple-output (MIMO) technique, the number of antennas at the BS is becoming large as well. In this case, the joint device activity detection and channel estimation problem involves a vast number of unknown variables. It is thus important to propose some efficient algorithm that can be implemented in a practical but large IoT system.

To tackle the above two challenges, this paper aims to propose a low-complexity algorithm to solve the problem of detecting the active devices and estimating their delay and channels in asynchronous mMTC systems. Specifically, by introducing an enlarged sensing matrix that consists of all the effective preambles of the devices (for each device, each of its effective preambles denotes a preamble that is delayed by a particular number of symbols), we show that the above problem can be cast into a compressed sensing problem, similar to its counterpart in synchronous mMTC systems \cite{Part1,Liu18,Chen18,senel2018grant}. To guarantee that at most one effective preamble is detected to be active for each device, the compressed sensing problem is further formulated as a group LASSO problem \cite{Yuan06}. We propose a block coordinate descent (BCD) algorithm to solve this problem globally. Thanks to the problem's special structure, we show that when the BCD algorithm optimizes some block of variables with the other blocks of variables being fixed, the optimal solution can be obtained in closed form. Further, the overall complexity of our algorithm is linear to the numbers of devices and antennas at the BS, which makes it appealing in large IoT systems. Last, we remark that our considered device activity detection problem in asynchronous mMTC systems is in sharp contrast to the information decoding counterpart in asynchronous human-type communication systems, which has been widely studied in the literature, because of the different techniques used for activity detection and information decoding.

\section{System Model}\label{sec:system model}

This paper considers the uplink communication in an mMTC system, which consists of a BS equipped with $M$ antennas and $N$ single-antenna IoT devices denoted by the set $\mathcal{N}\triangleq\{1,\ldots,N\}$. We assume quasi-static block-fading channels, in which all channels remain approximately constant in each fading block, but vary independently from block to block. The channel from device $n$ to the BS is denoted by $\mv{h}_n\in \mathbb{C}^{M\times 1}$, $\forall n\in \mathcal{N}$. We assume that the device channels follow the independent and identically distributed (i.i.d.) Rayleigh fading channel model, i.e., $\mv{h}_n\sim \mathcal{CN}(\mv{0},\alpha_n\mv{I})$, $\forall n$, where $\alpha_n$ denotes the path loss of device $n$.

Due to the sporadic IoT data traffic, only some of the devices become active within each channel coherence block. We thus define the device activity indicator functions as follows:
\begin{align}\label{eqn:indicator}
\lambda_n=\left\{\begin{array}{ll} 1, & {\rm if} ~ {\rm device } ~ n ~ {\rm is ~ active}, \\ 0, & {\rm otherwise}, \end{array}\right. ~ \forall n\in \mathcal{N}.
\end{align}Then, the set of active devices is defined by $\mathcal{K}=\{n:\lambda_n=1,\forall n\in \mathcal{N}\}$.

It is assumed that the two-phase grant-free random access scheme \cite{Liu18} is adopted for the considered system, where each active device first sends its preamble sequence to the BS for device activity detection and channel estimation, and then sends its data to the BS for decoding. Further, this paper mainly focuses on the first phase of the above grant-free random access scheme. In this phase, each user $n$ is assigned with a unique preamble sequence denoted by $\mv{a}_n=[a_{n,1},\ldots,a_{n,L}]^T\in \mathbb{C}^{L\times 1}$, $\forall n$, where $L$ denotes the length of the preamble sequence, and $a_{n,l}$ with $|a_{n,l}|=1$ denotes the $l$-th preamble symbol of device $n$, $l=1,\ldots,L$.

At the beginning of each coherence block, all the active users tend to transmit their preambles to the BS. However, due to the lack of perfect coordination among a large number of low-cost devices, the preamble transmissions are in general asynchronous. Let $\tau_n$ denote the discrete delay (in terms of symbols) of user $n$ for transmitting $\mv{a}_n$, $\forall n\in \mathcal{K}$. It is assumed that at each coherence block, the preamble transmission delay for each device $n$ is a random integer value in the regime of $\tau_n\in [0,\tau_{{\rm max}}]$, $\forall n$, where $\tau_{{\rm max}}$ denotes the maximum delay of all the devices over all the coherence blocks. Moreover, $\tau_{{\rm max}}$ is assumed to be known by the BS.

Since each active device $n\in \mathcal{K}$ starts to transmit its preamble at the $(\tau_n+1)$-th symbol in the coherence block, we define the effective preamble sequence for device $n$ as
\begin{align}\label{eqn:effective preamble sequence}
\bar{\mv{a}}_n(\tau_n)&=[\bar{a}_{n,1}(\tau_n),\ldots,\bar{a}_{n,L+\tau_{{\rm max}}}(\tau_n)]^T \nonumber \\ &=[\underset{\tau_n}{\underbrace{0,\ldots,0}},a_{n,1},\ldots,a_{n,L},\underset{\tau_{{\rm max}}-\tau_n}{\underbrace{0,\ldots,0}}]^T, ~ \forall n,
\end{align}where $\bar{a}_{n,j}(\tau_n)$ denotes the $j$-th transmit symbol of device $n$ given a delay of $\tau_n$ symbols, $1\leq j \leq L+\tau_{{\rm max}}$. Note that $\bar{a}_{n,j}=0$ if $1\leq j\leq \tau_n$ or $L+\tau_n+1\leq j\leq L+\tau_{{\rm max}}$. Moreover, under the two-phase grant-free random access scheme, after transmitting the pilot in Phase I, each active device should wait for $\tau_{{\rm max}}$ symbols before transmitting its data such that the pilot received at the BS in the first $L+\tau_{{\rm max}}$ time slots is not interfered by the data. Then, at time slot $1\leq j \leq L+\tau_{{\rm max}}$, the received signal at the BS is merely contributed by pilot and given by
\begin{align}\label{eqn:signal at BS}
\mv{y}_{j}\hspace{-2pt}=\hspace{-2pt}\sum\limits_{n=1}^N\lambda_n\mv{h}_n \sqrt{p}\bar{a}_{n,j}(\tau_n)\hspace{-2pt}+\hspace{-2pt}\mv{z}_{j}, ~ j=1,\ldots,L+\tau_{{\rm max}},
\end{align}where $p$ denotes the identical transmit power of all the devices, and $\mv{z}_{j}\sim \mathcal{CN}(\mv{0},\sigma^2\mv{I})$ denotes the additive white Gaussian noise (AWGN) at the BS. Further, the overall received signal at the BS over all the $L+\tau_{{\rm max}}$ time slots, denoted by $\mv{Y}=[\mv{y}_1,\ldots,\mv{y}_{L+\tau_{{\rm max}}}]^T\in \mathbb{C}^{(L+\tau_{{\rm max}})\times M}$, is given by
\begin{align}\label{eqn:received signal}
\mv{Y}\hspace{-2pt}=\hspace{-2pt}\sum\limits_{n=1}^N\lambda_n\hspace{-1pt} \sqrt{p}\bar{\mv{a}}_n(\tau_n) \mv{h}_n^T\hspace{-3pt}+\hspace{-2pt}\mv{Z}\hspace{-2pt}=\hspace{-4pt}\sqrt{p}\bar{\mv{A}}(\tau_1,\hspace{-2pt}\ldots\hspace{-2pt},\tau_N)\mv{X}\hspace{-2pt}+\hspace{-2pt}\mv{Z},
\end{align}where $\bar{\mv{A}}(\tau_1,\ldots,\tau_N)=[\bar{\mv{a}}_1(\tau_1),\ldots,\bar{\mv{a}}_N(\tau_N)]$, $\mv{X}=[\mv{x}_1,$ $\ldots,\mv{x}_N]^T$ with $\mv{x}_n=\lambda_n\mv{h}_n$, $\forall n$, and $\mv{Z}=[\mv{z}_1,\ldots,\mv{z}_{L+\tau_{{\rm max}}}]^T$. The goal of the BS is to estimate the device activity $\lambda_n$'s and delay $\tau_n$'s as well as active devices' channels $\mv{h}_n$'s based on its received signal $\mv{Y}$ given in (\ref{eqn:received signal}) and its knowledge of the preamble sequences $\mv{a}_n$'s.

\section{A Compressed Sensing Problem Formulation for Estimating \lowercase{$\{\lambda_n,\tau_n,\mv{h}_n\}$}}\label{sec:Device Activity and Delay Detection as a Compressed Sensing Problem}

In this section, we show that the detection of active devices as well as the estimation of their delay and channels can be cast into a compressed sensing problem. Specifically, define
\begin{align}\label{eqn:sensing matrix}
\mv{A}_{{\rm ext}}&=[\bar{\mv{a}}_1(0), \ldots, \bar{\mv{a}}_1(\tau_{{\rm max}}), \ldots, \bar{\mv{a}}_N(0), \ldots, \bar{\mv{a}}_N(\tau_{{\rm max}})]\nonumber \\ &\in \mathbb{C}^{(L+\tau_{{\rm max}})\times (\tau_{{\rm max}}+1)N}
\end{align}as the collection of the possible effective preamble sequences of all the devices. Moreover, define the indicator functions of device activity and delay as follows:
\begin{align}\label{eqn:indicator activity and delay}
\beta_{n,\tau}=\left\{\begin{array}{ll} 1, & {\rm if} ~ \lambda_n=1 ~ {\rm and} ~ \tau=\tau_n, \\ 0, & {\rm otherwise}, \end{array}\right. ~ \forall n, \tau.
\end{align}In other words, $\beta_{n,\tau}=1$ only if device $n$ is active and its delay is of $\tau$ symbol duration. Note that if device $n$ is active, only one of $\beta_{n,\tau}$'s, $\tau=0,\ldots,\tau_{{\rm max}}$, is equal to 1, i.e., $\sum_{\tau=0}^{\tau_{{\rm max}}}\beta_{n,\tau}\leq 1$, $\forall n$. Then, (\ref{eqn:received signal}) can be reformulated as
\begin{align}\label{eqn:received signal new}
\mv{Y}=\sqrt{p}\mv{A}^{{\rm ext}}\mv{X}^{{\rm ext}}+\mv{Z},
\end{align}where $\mv{X}^{{\rm ext}}= [\mv{x}_{1,0}^{{\rm ext}},\ldots,\mv{x}_{1,\tau_{{\rm max}}}^{{\rm ext}},\ldots,\mv{x}_{N,0}^{{\rm ext}},\ldots,\mv{x}_{N,\tau_{{\rm max}}}^{{\rm ext}}]^T\in \mathbb{C}^{(\tau_{{\rm max}}+1)N\times M}$ with
\begin{align}\label{eqn:detect}
\mv{x}_{n,\tau}^{{\rm ext}}=\beta_{n,\tau}\mv{h}_n, ~~~ \forall n,\tau.
\end{align}

Suppose that $\mv{X}^{{\rm ext}}$ can be estimated according to (\ref{eqn:received signal new}). If $\mv{x}_{n,\tau}^{{\rm ext}}\neq \mv{0}$, i.e., $\beta_{n,\tau}=1$, device $n$ is active, i.e., $\lambda_n=1$, and its delay and channel are $\tau_n=\tau$ and $\mv{h}_n=\mv{x}_{n,\tau}^{{\rm ext}}$. If $\mv{x}_{n,\tau}^{{\rm ext}}=\mv{0}$, $\forall \tau$, i.e., $\beta_{n,\tau}=0$, $\forall \tau$, device $n$ is inactive, i.e., $\lambda_n=0$. Thus, the key of the joint estimation of device activity, delay, and channels lies in estimating $\mv{X}^{{\rm ext}}$ based on (\ref{eqn:received signal new}).

Note that estimating $\mv{X}^{{\rm ext}}$ based on (\ref{eqn:received signal new}) is a compressed sensing problem, since $\mv{X}^{{\rm ext}}$ is a row-sparse matrix according to (\ref{eqn:indicator activity and delay}) and (\ref{eqn:detect}). In this paper, the compressed sensing problem of estimating $\mv{X}^{{\rm ext}}$ is formulated as follows:
\begin{align} \mathop{\mathrm{minimize}}_{\mv{X}^{{\rm ext}}} & ~ \left\|\mv{Y}-\sqrt{p}\mv{A}^{{\rm ext}}\mv{X}^{{\rm ext}}\right\|_{{\rm F}}^2 \label{eqn:problem} \\
\mathrm {subject ~ to}  & ~ \left\|\left[\left\|\mv{x}_{n,0}^{{\rm ext}}\right\|_2,\ldots, \left\|\mv{x}_{n,\tau_{{\rm max}}}^{{\rm ext}}\right\|_2\right]\right\|_0\leq 1, ~~~ \forall n, \label{eqn:sparse constraint}
\end{align}where $\|\mv{A}\|_{{\rm F}}$ denotes the Frobenius norm of matrix $\mv{A}$, i.e., $\|\mv{A}\|_{{\rm F}}=\sqrt{{\rm tr}(\mv{A}\mv{A}^H)}$, and $\|\mv{a}\|_0$ denotes the zero norm of vector $\mv{a}$, i.e., the number of non-zero elements in $\mv{a}$.

In the above problem, constraint (\ref{eqn:sparse constraint}) is to guarantee that at most one delay pattern is detected to be active for each device $n$. Mathematically, (\ref{eqn:sparse constraint}) imposes a group sparsity constraint on the structure of $\mv{X}^{{\rm ext}}$: in each block consisting of $\tau_{{\rm max}}+1$ vectors $\mv{x}_{n,0}^{{\rm ext}},\ldots,\mv{x}_{n,\tau_{{\rm max}}}^{{\rm ext}}$, at most one of them is a non-zero vector. However, this constraint is non-convex. In the rest of this paper, we adopt the group LASSO technique to deal with this non-convex group sparsity constraint.

Under the group LASSO technique, given any coefficient $\rho>0$, we are interested in the following convex problem \cite{Yuan06}
\begin{align} \hspace{-8pt} \mathop{\mathrm{minimize}}_{\mv{X}^{{\rm ext}}}  ~ 0.5\left\|\hspace{-1pt}\mv{Y}\hspace{-3pt}-\hspace{-3pt}\sqrt{p}\mv{A}^{{\rm ext}}\mv{X}^{{\rm ext}}\hspace{-1pt}\right\|_{{\rm F}}^2\hspace{-2pt}+\hspace{-2pt}\rho\hspace{-1pt}\sum\limits_{n=1}^N\hspace{-1pt}\sum\limits_{\tau=0}^{\tau_{{\rm max}}}\hspace{-1pt}\left\|\mv{x}_{n,\tau}^{{\rm ext}}\right\|_2. \label{eqn:problem 1}
\end{align}In problem (\ref{eqn:problem 1}), we penalize the estimation error with a mixed $\ell_2/\ell_1$ norm, i.e., $\rho\sum_{n=1}^N\sum_{\tau=0}^{\tau_{{\rm max}}}\|\mv{x}_{n,\tau}^{{\rm ext}}\|_2$. Note that this penalty is minimized when all the zero entries are put together in some rows of $\mv{X}^{{\rm ext}}$. As a result, given a large value of $\rho$, the optimal solution of problem (\ref{eqn:problem 1}) should be a row-sparse matrix. Moreover, if $\rho$ is large enough, the corresponding solution will be sufficiently sparse, and therefore, constraint (\ref{eqn:sparse constraint}) in problem (\ref{eqn:problem}) can be satisfied.

In the following two sections, we will introduce how to solve problem (\ref{eqn:problem 1}) efficiently given $\rho$ and how to select a proper value of $\rho$ so as to balance between activity sparsity and channel estimation error, respectively.

\section{An Efficient BCD Algorithm for Problem (11)}\label{sec:Standard BCD Algorithm}
The BCD type of algorithms are efficient in solving large-scale optimization problems with a vast number of variables \cite{ber99programming}. In this section, we introduce a low-complexity BCD algorithm to solve problem (\ref{eqn:problem 1}) given any $\rho>0$.

\subsection{Algorithm Design}
Under the BCD algorithm, at each time, we merely optimize one vector $\mv{x}_{\bar{n},\bar{\tau}}^{{\rm ext}}$ for some particular $1\leq \bar{n}\leq N$ and $0\leq \bar{\tau}\leq \tau_{{\rm max}}$ given $\mv{x}_{n,\tau}^{{\rm ext}}=\tilde{\mv{x}}_{n,\tau}^{{\rm ext}}$'s, $\forall (n,\tau)\neq (\bar{n},\bar{\tau})$. The corresponding optimization problem is formulated as
\begin{align} \mathop{\mathrm{minimize}}_{\mv{x}_{\bar{n},\bar{\tau}}^{{\rm ext}}} & ~ 0.5\left\|\tilde{\mv{Y}}_{\bar{n},\bar{\tau}}\hspace{-2pt}-\hspace{-2pt}\sqrt{p}\mv{a}_{\bar{n},\bar{\tau}}^{{\rm ext}}\left(\mv{x}_{\bar{n},\bar{\tau}}^{{\rm ext}}\right)^T\right\|_{{\rm F}}^2\hspace{-2pt}+\hspace{-2pt}\rho\left\|\mv{x}_{\bar{n},\bar{\tau}}^{{\rm ext}}\right\|_2 \label{eqn:problem 2}
\end{align}where
\begin{align}\label{eqn:A}
\tilde{\mv{Y}}_{\bar{n},\bar{\tau}}=\mv{Y}-\sqrt{p}\sum\limits_{(n,\tau)\neq (\bar{n},\bar{\tau})}\mv{a}_{n,\tau}^{{\rm ext}}\left(\tilde{\mv{x}}_{n,\tau}^{{\rm ext}}\right)^T.
\end{align}Somewhat surprisingly, we can obtain the closed-form optimal solution of problem (\ref{eqn:problem 2}), as shown in the following theorem.

\begin{theorem}\label{theorem1}
The objective function in problem (\ref{eqn:problem 2}) is strongly convex, and its global minimum is achieved by
\begin{align}\label{eqn:optimal solution}
\left(\hat{\mv{x}}_{\bar{n},\bar{\tau}}^{{\rm ext}}\right)^T\hspace{-3pt}=\hspace{-3pt}\left\{\begin{array}{ll}\hspace{-3pt} \gamma_{\bar{n},\bar{\tau}} \left(\mv{a}_{\bar{n},\bar{\tau}}^{{\rm ext}}\right)^H\tilde{\mv{Y}}_{\bar{n},\bar{\tau}}, \hspace{-2pt} & \hspace{-2pt} {\rm if} ~ \left\|\left(\mv{a}_{\bar{n},\bar{\tau}}^{{\rm ext}}\right)^H\tilde{\mv{Y}}_{\bar{n},\bar{\tau}}\right\|_2\hspace{-2pt}>\hspace{-2pt}\frac{\rho}{\sqrt{p}}, \\
\hspace{-3pt} \mv{0}, \hspace{-2pt} & \hspace{-2pt} {\rm otherwise},\end{array}\right.
\end{align}where
\begin{align}\label{eqn:gamma}
\gamma_{\bar{n},\bar{\tau}}=\frac{1}{L\sqrt{p}}-\frac{\rho}{Lp\left\|\left(\mv{a}_{\bar{n},\bar{\tau}}^{{\rm ext}}\right)^H\tilde{\mv{Y}}_{\bar{n},\bar{\tau}}\right\|_2}>0.
\end{align}
\end{theorem}

\begin{proof}
Please refer to Appendix A.
\end{proof}

The optimal solution (\ref{eqn:optimal solution}) in Theorem \ref{theorem1} indicates that the BS should keep applying the matched filters $\mv{a}_{\bar{n},\bar{\tau}}^{{\rm ext}}$'s to denoise $\tilde{\mv{Y}}_{\bar{n},\bar{\tau}}$'s. Then, if the power of some resulting signal, i.e., $\left\|\left(\mv{a}_{\bar{n},\bar{\tau}}^{{\rm ext}}\right)^H\tilde{\mv{Y}}_{\bar{n},\bar{\tau}}\right\|_2$, is larger than a threshold $\rho/\sqrt{p}$, then the estimation of $\mv{x}_{\bar{n},\bar{\tau}}^{{\rm ext}}$ is a non-zero vector. Otherwise, the estimation of $\mv{x}_{\bar{n},\bar{\tau}}^{{\rm ext}}$ is a zero vector. This implies that the solution to problem (\ref{eqn:problem 1}) is more sparse as $\rho$ increases.

\begin{remark}
In general, a group LASSO problem can be merely solved numerically. However, under the BCD framework, the sensing matrix reduces to a vector $\mv{a}_{\bar{n},\bar{\tau}}^{{\rm ext}}$ in problem (\ref{eqn:problem 2}) to optimize $\mv{x}_{\bar{n},\bar{\tau}}^{{\rm ext}}$. In this case, there is a closed-form solution, which is appealing to reduce the computational complexity in mMTC.
\end{remark}

Based on Theorem \ref{theorem1}, the BCD algorithm to solve problem (\ref{eqn:problem 1}) is summarized in Algorithm \ref{table2}. Algorithm \ref{table2} is an iterative algorithm. At each outer iteration of the algorithm, we first optimize $\mv{x}_{1,0}^{{\rm ext}}$ given $\mv{x}_{n,\tau}=\tilde{\mv{x}}_{n,\tau}$'s, $\forall (n,\tau)\neq (1,0)$, and then optimize $\mv{x}_{1,1}^{{\rm ext}}$ given $\mv{x}_{n,\tau}=\tilde{\mv{x}}_{n,\tau}$'s, $\forall (n,\tau)\neq (1,1)$, and so on, as shown in Step 2.1 to Step 2.5. When $\mv{x}_{n,\tau}^{{\rm ext}}$'s are all optimized once, we can calculate the objective value of problem (\ref{eqn:problem 1}) achieved after the $t$-th iteration of the algorithm, denoted by $\Gamma^{(t)}$ as shown in Step 3. The algorithm terminates when the objective value of problem (\ref{eqn:problem 1}) does not decrease sufficiently over two iterations.

\begin{algorithm}[t]
{\bf Initialization}: Set the initial values of $\mv{x}_{n,\tau}^{{\rm ext}}$'s as $\tilde{\mv{x}}_{n,\tau}^{{\rm ext}}=\mv{0}$, $n=1,\ldots,N$, $\tau=0,\ldots,\tau_{{\rm max}}$, $\bar{\mv{Y}}=\mv{Y}$, where $\mv{Y}$ is the received signal given in (\ref{eqn:received signal}), and $t=1$; \\
{\bf Repeat}:
\begin{itemize}
\item[1] Set $j=1$;
\item[2] While $j\leq N(\tau_{{\rm max}}+1)$:
\begin{itemize}
\item[2.1] Set $\bar{n}=\left\lfloor j/(\tau_{{\rm max}}+1)\right\rfloor+1$ and $\bar{\tau}=j-(\bar{n}-1)(\tau_{{\rm max}}+1)$;
\item[2.2] Set $\mv{x}_{n,\tau}^{{\rm ext}}=\tilde{\mv{x}}_{n,\tau}^{{\rm ext}}$'s, $\forall (n,\tau)\neq (\bar{n},\bar{\tau})$, and $\tilde{\mv{Y}}_{\bar{n},\bar{\tau}}=\bar{\mv{Y}}+\sqrt{p}\mv{a}_{\bar{n},\bar{\tau}}^{{\rm ext}}\left(\tilde{\mv{x}}_{\bar{n},\bar{\tau}}^{{\rm ext}}\right)^T$;
\item[2.3] Find the optimal solution to problem (\ref{eqn:problem 2}) based on (\ref{eqn:optimal solution}), which is denoted by $\hat{\mv{x}}_{\bar{n},\bar{\tau}}^{{\rm ext}}$;
\item[2.4] Set $\tilde{\mv{x}}_{\bar{n},\bar{\tau}}^{{\rm ext}}=\hat{\mv{x}}_{\bar{n},\bar{\tau}}^{{\rm ext}}$ and $\bar{\mv{Y}}=\tilde{\mv{Y}}_{\bar{n},\bar{\tau}}-\sqrt{p}\mv{a}_{\bar{n},\bar{\tau}}^{{\rm ext}}\left(\hat{\mv{x}}_{\bar{n},\bar{\tau}}^{{\rm ext}}\right)^T$;
\item[2.5] Set $j=j+1$;
\end{itemize}
\item[3] Set $\Gamma^{(t)}=0.5\left\|\bar{\mv{Y}}\right\|_{{\rm F}}^2+\rho\sum_{n=1}^N\sum_{\tau=0}^{\tau_{{\rm max}}}\left\|\tilde{\mv{x}}_{n,\tau}^{{\rm ext}}\right\|_2$ and $t=t+1$;
\end{itemize}
{\bf Until} $(\Gamma^{(t-1)}-\Gamma^{(t)})/\Gamma^{(t-1)}\leq \xi$, where $\xi$ is a small positive number.
\caption{BCD Algorithm for Solving Problem (\ref{eqn:problem 1})}
\label{table2}
\end{algorithm}

\subsection{Algorithm Properties}
After introducing how the BCD algorithm works, in this subsection, we present some theoretical properties of this algorithm about its optimality and complexity.
\begin{theorem}\label{theorem2}
Every limit point of the iterates generated by Algorithm \ref{table2} is a global solution of problem (\ref{eqn:problem 1}). Moreover, for all sufficiently large $t$, it follows that
\begin{align}\label{eqn:sublinear}
\Gamma^{(t)}-\Gamma^\ast\leq {\cal O}\left(\frac{1}{t}\right),
\end{align}where $\Gamma^{(t)}$, as given in Step 3 of Algorithm \ref{table2}, is the objective value of problem (\ref{eqn:problem 1}) at the $t$-th iteration of Algorithm \ref{table2}, and $\Gamma^\ast$ is the optima value of problem (\ref{eqn:problem 1}).
\end{theorem}

\begin{proof}
Please refer to Appendix B.		
\end{proof}

\begin{theorem}\label{theorem3}
Given any $\epsilon>0$, the total complexity of Algorithm \ref{table2} to find an $\epsilon$-optimal solution of problem (\ref{eqn:problem 1}) that satisfies $\Gamma^{(t)}-\Gamma^\ast\leq \epsilon$ is given by \begin{align}\label{eqn:complexity}
	{\cal O}\left(\frac{(L+\tau_{\max})\tau_{\max}MN}{\epsilon}\right).
	\end{align}
\end{theorem}

\begin{proof}
Please refer to Appendix C.
\end{proof}

Theorems \ref{theorem2} and \ref{theorem3} imply that the BCD Algorithm can solve problem (\ref{eqn:problem 1}) globally with a complexity linear to $N$ and $M$.

\section{The Approach to Select $\rho$}\label{sec:The Approach to Select rho}
After solving problem (\ref{eqn:problem 1}) with any given $\rho>0$, we introduce in this section how to determine the value of $\rho$ such that the solution to problem (\ref{eqn:problem 1}) is a good solution to problem (\ref{eqn:problem}). In this paper, we update the value of $\rho$ iteratively. Specifically, at the beginning, we set an initial value to $\rho$ as $\rho=\rho^{{\rm initial}}>0$ and solve problem (\ref{eqn:problem 1}) via Algorithm \ref{table2}. Then, we keep updating $\rho$ as $\rho=\delta \rho$, where $\delta>1$, and solving problem (\ref{eqn:problem 1}) iteratively until for some large enough value of $\rho$, the solution to problem (\ref{eqn:problem 1}) satisfies constraint (\ref{eqn:sparse constraint}) in problem (\ref{eqn:problem}). The overall algorithm to solve problem (\ref{eqn:problem}) via solving a sequence of problem (\ref{eqn:problem 1}) is summarized in Algorithm \ref{table1}.
\begin{algorithm}[t]
{\bf Initialization}: Set an initial value of $\rho$ as $\rho=\rho^{{\rm initial}}>0$; \\
{\bf Repeat}:
\begin{itemize}
\item[1] Find the solution to problem (\ref{eqn:problem 1}) given $\rho$, which is denoted by $\tilde{\mv{x}}_{n,\tau}^{{\rm ext}}$'s, via Algorithm \ref{table2};
\item[2] Set
\begin{align}\label{eqn:sparse matrix 1}
\mv{x}_{n,\tau}^{{\rm ext}}=\left\{\begin{array}{ll} \tilde{\mv{x}}_{n,\tau}^{{\rm ext}}, & {\rm if} ~ \frac{\left\|\tilde{\mv{x}}_{n,\tau}^{{\rm ext}}\right\|_2^2}{M}\geq \zeta \alpha_n, \\ \mv{0}, & {\rm otherwise}, \end{array}\right. ~ \forall n,\tau,
\end{align}where $0<\zeta<1$ is a given parameter to control the sparsity of $\mv{X}^{{\rm ext}}$.
\item[3] Update $\rho=\delta\rho$ where $\delta>1$.
\end{itemize}
{\bf Until} the solution of $\mv{x}_{n,\tau}^{{\rm ext}}$'s satisfies constraint (\ref{eqn:sparse constraint}).
\caption{Proposed Algorithm for Solving Problem (\ref{eqn:problem})}
\label{table1}
\end{algorithm}

Note that for some inactive devices, maybe the power of corresponding estimated signals $\tilde{\mv{x}}_{n,\tau}^{{\rm ext}}$'s via Algorithm \ref{table2} are very weak, but non-zero. This will cause the so-called false alarm event, i.e., an inactive device is detected as an active device. To enhance the sparsity of the estimation of $\mv{X}^{{\rm ext}}$ and reduce the false alarm probability, after the convergence of Algorithm \ref{table2} in Step 1 of Algorithm \ref{table1}, we set $\mv{x}_{n,\tau}^{{\rm ext}}=\mv{0}$ if $\|\tilde{\mv{x}}_{n,\tau}^{{\rm ext}}\|^2$ is less than some threshold in Step 2 of Algorithm \ref{table1}.

\begin{remark}
One main issue under the conventional LASSO technique is how to select a proper value of $\rho$ such that the resulting LASSO problem is a good approximation of the original problem. Under our interested problem (\ref{eqn:problem}), the new constraint (\ref{eqn:sparse constraint}) enables an accurate stopping criterion for updating $\rho$ in Step 3 of Algorithm \ref{table1}. This is another advantage to use LASSO in this work, other than the closed-form solution given any $\rho$ shown in Theorem \ref{theorem1}.
\end{remark}

\section{Numerical Results}\label{sec:Numerical Results}
In this section, we provide numerical examples to verify the effectiveness of our proposed algorithm for detecting the active devices and estimating their delay and channels in asynchronous IoT systems. We assume that there are $N=100$ IoT devices located in a cell with a radius of $250$ m, and at each coherence block, only $K=10$ of them become active. Moreover, the maximum delay of all the devices is $\tau_{{\rm max}}=5$ symbols. The transmit power of the active devices is $23$ dBm. Last, the power spectral density of the AWGN at the BS is $-169$dBm/Hz, and the channel bandwidth is $10$ MHz.
\begin{figure}[t]
  \centering
  \includegraphics[width=8cm]{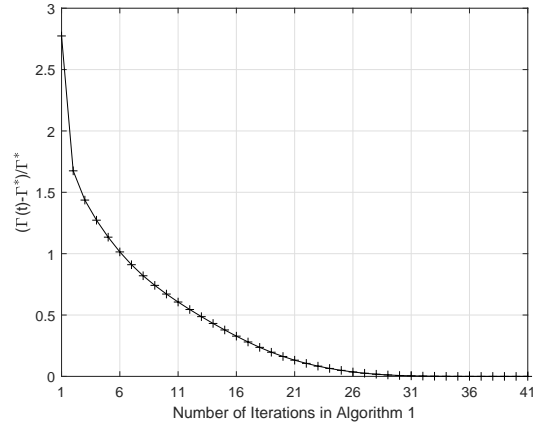}\vspace{-5pt}
  \caption{Sublinear convergence rate of Algorithm \ref{table2}.}\label{fig2}
\end{figure}

First, we provide one numerical example to verify the convergence property of Algorithm \ref{table2}, where the BS has $M=128$ antennas, and the pilot sequence length is $L=20$. Fig. \ref{fig2} shows the relative gap between the objective value of problem (\ref{eqn:problem 1}) achieved at each iteration of Algorithm \ref{table2} and the optimal value of problem (\ref{eqn:problem 1}), i.e., $(\Gamma^{(t)}-\Gamma^\ast)/\Gamma^\ast$. As shown in Theorem \ref{theorem2}, the solution generated by Algorithm \ref{table2} converges to the optimal solution sublinearly.

\begin{figure}[t]
  \centering
  \includegraphics[width=8cm]{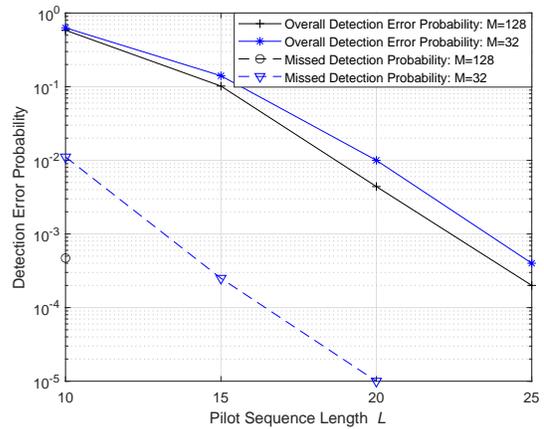}\vspace{-5pt}
  \caption{Detection error probability versus $L$.}\label{fig3}\vspace{-15pt}
\end{figure}

Next, we show the performance of the proposed algorithm by Monte Carlo simulation. Specifically, we generate $10^4$ realizations of device activity, location, and channels. Moreover, for each realization, if the detection of some $\beta_{k,n}$ is wrong, then we declare that this realization is under detection error. The overall detection error probability is defined as the ratio between the number of realizations under detection error and the total number of realizations, i.e., $10^4$. Similar to \cite{Part1}, the missed detection/false alarm probability is defined as the probability that an active/inactive device is detected as an inactive/active device. Fig. \ref{fig3} shows the overall detection error probability and missed detection probability (no false alarm events happen over the $10^4$ realizations) achieved by our proposed algorithm when $L$ ranges from 10 to 25 and $M=128$ or $32$. First, it is observed that the missed detection and false alarm probabilities for device activity detection are very low, e.g., when $M=128$ and $L\geq 15$, no missed detection and false alarm events are observed over the $10^4$ realizations. Next, it is observed that when $L$ is small, the overall detection error probability is high. This indicates that although the active devices can be detected, their delay estimation is in error with high probability. However, the priority of device activity detection is much higher than that of delay estimation. Moreover, delay estimation becomes more accurate as $L$ increases. Last, it is observed that massive MIMO is powerful to decrease the detection error probability.

\section{Conclusion}\label{sec:Conclusion}
In this paper, we showed that the problem of jointly detecting the active devices and estimating their delay and channels in asynchronous mMTC systems can be formulated as a group LASSO problem. Utilizing the BCD technique, we proposed an efficient algorithm to solve the group LASSO problem, the complexity of which is shown to be linearly proportional to the numbers of devices and antennas at the BS. Future work may consider how to apply the covariance-based device detection appraoch \cite{haghighatshoar2018improved,Chen19} in asynchronous IoT systems.

\section*{Appendix}

\subsection*{A: Proof of Theorem 1}\label{appendix1}
Given any $\bar{n}=1,\ldots,N$ and $\bar{\tau}=0,\ldots,\tau_{{\rm max}}$, define
\begin{align}
f_{\bar{n},\bar{\tau}}(\hspace{-1pt}\mv{x}_{\bar{n},\bar{\tau}}^{{\rm ext}}\hspace{-1pt})\hspace{-2pt}=\hspace{-2pt}0.5\left\|\hspace{-1pt}\tilde{\mv{Y}}_{\bar{n},\bar{\tau}}\hspace{-3pt}-\hspace{-3pt}\sqrt{p}\mv{a}_{\bar{n},\bar{\tau}}^{{\rm ext}}\left(\hspace{-1pt}\mv{x}_{\bar{n},\bar{\tau}}^{{\rm ext}}\hspace{-1pt}\right)^T\hspace{-1pt}\right\|_{{\rm F}}^2\hspace{-3pt}+\hspace{-3pt}\rho\left\|\hspace{-1pt}\mv{x}_{\bar{n},\bar{\tau}}^{{\rm ext}}\hspace{-1pt}\right\|_2.
\end{align}Since $\|\tilde{\mv{Y}}_{\bar{n},\bar{\tau}}-\sqrt{p}\mv{a}_{\bar{n},\bar{\tau}}^{{\rm ext}}\left(\mv{x}_{\bar{n},\bar{\tau}}^{{\rm ext}}\right)^T\|_{{\rm F}}^2$ is a strongly convex function and $\rho\|\mv{x}_{\bar{n},\bar{\tau}}^{{\rm ext}}\|_2$ is a convex function, $f_{\bar{n},\bar{\tau}}(\hspace{-1pt}\mv{x}_{\bar{n},\bar{\tau}}^{{\rm ext}})$ is a strongly convex function.

Next, we derive the optimal solution to minimize $f_{\bar{n},\bar{\tau}}(\hspace{-1pt}\mv{x}_{\bar{n},\bar{\tau}}^{{\rm ext}})$. It is observed that $f_{\bar{n},\bar{\tau}}(\mv{x}_{\bar{n},\bar{\tau}}^{{\rm ext}})$ is differentiable when $\mv{x}_{\bar{n},\bar{\tau}}^{{\rm ext}} \neq \mv{0}$, but not differentiable when $\mv{x}_{\bar{n},\bar{\tau}}^{{\rm ext}}=\mv{0}$. Moreover, when $\mv{x}_{\bar{n},\bar{\tau}}^{{\rm ext}}\neq \mv{0}$, the gradient of $f_{\bar{n},\bar{\tau}}(\mv{x}_{\bar{n},\bar{\tau}}^{{\rm ext}})$ is
\begin{align}\label{eqn:subgradient 1}
& [\partial f_{\bar{n},\bar{\tau}}(\mv{x}_{\bar{n},\bar{\tau}}^{{\rm ext}})]^T \nonumber \\ =&-\sqrt{p}\left(\mv{a}_{\bar{n},\bar{\tau}}^{{\rm ext}}\right)^H \tilde{\mv{Y}}_{\bar{n},\bar{\tau}}+pL\left(\mv{x}_{\bar{n},\bar{\tau}}^{{\rm ext}}\right)^T+\frac{\rho \left(\mv{x}_{\bar{n},\bar{\tau}}^{{\rm ext}}\right)^T}{\left\|\mv{x}_{\bar{n},\bar{\tau}}^{{\rm ext}}\right\|_2},
\end{align}while when $\mv{x}_{\bar{n},\bar{\tau}}^{{\rm ext}}=\mv{0}$, the sub-gradient of $f_{\bar{n},\bar{\tau}}(\mv{x}_{\bar{n},\bar{\tau}}^{{\rm ext}})$ is
\begin{align}\label{eqn:subgradient 2}
& [\partial f_{\bar{n},\bar{\tau}}(\mv{x}_{\bar{n},\bar{\tau}}^{{\rm ext}})]^T \nonumber \\ =&-\sqrt{p}\left(\mv{a}_{\bar{n},\bar{\tau}}^{{\rm ext}}\right)^H \tilde{\mv{Y}}_{\bar{n},\bar{\tau}}+pL\left(\mv{x}_{\bar{n},\bar{\tau}}^{{\rm ext}}\right)^T+\rho\partial \left\|\mv{x}_{\bar{n},\bar{\tau}}^{{\rm ext}}\right\|_2.
\end{align}

Since $f_{\bar{n},\bar{\tau}}(\mv{x}_{\bar{n},\bar{\tau}}^{{\rm ext}})$ is a strongly convex function over $\mv{x}_{\bar{n},\bar{\tau}}^{{\rm ext}}$, a point $\hat{\mv{x}}_{\bar{n},\bar{\tau}}^{{\rm ext}}$ minimizes this function if and only if $\mv{0}$ is a sub-gradient of the function $f_{\bar{n},\bar{\tau}}(\mv{x}_{\bar{n},\bar{\tau}}^{{\rm ext}})$ at this point, i.e.,
\begin{align}
\mv{0}\in \partial f_{\bar{n},\bar{\tau}}(\mv{x}_{\bar{n},\bar{\tau}}^{{\rm ext}})|_{\mv{x}_{\bar{n},\bar{\tau}}^{{\rm ext}}=\hat{\mv{x}}_{\bar{n},\bar{\tau}}^{{\rm ext}}}.
\end{align}According to (\ref{eqn:subgradient 1}) and (\ref{eqn:subgradient 2}), we study the sub-gradient of the function $f_{\bar{n},\bar{\tau}}(\mv{x}_{\bar{n},\bar{\tau}}^{{\rm ext}})$ in two cases: $\mv{x}_{\bar{n},\bar{\tau}}^{{\rm ext}}\neq \mv{0}$ and $\mv{x}_{\bar{n},\bar{\tau}}^{{\rm ext}}=\mv{0}$.

First, consider the case when $\mv{x}_{\bar{n},\bar{\tau}}^{{\rm ext}}\neq \mv{0}$. To make $\partial f_{\bar{n},\bar{\tau}}(\mv{x}_{\bar{n},\bar{\tau}}^{{\rm ext}})=\mv{0}$, according to (\ref{eqn:subgradient 1}), we have
\begin{align}\label{eqn:gradient}
-\sqrt{p}\left(\mv{a}_{\bar{n},\bar{\tau}}^{{\rm ext}}\right)^H \tilde{\mv{Y}}_{\bar{n},\bar{\tau}}+pL\left(\mv{x}_{\bar{n},\bar{\tau}}^{{\rm ext}}\right)^T+\frac{\rho \left(\mv{x}_{\bar{n},\bar{\tau}}^{{\rm ext}}\right)^T}{\left\|\mv{x}_{\bar{n},\bar{\tau}}^{{\rm ext}}\right\|_2}=\mv{0}.
\end{align}It then follows that
\begin{align}\label{eqn:proportional}
\left(\mv{x}_{\bar{n},\bar{\tau}}^{{\rm ext}}\right)^T=\gamma_{\bar{n},\bar{\tau}}\left(\mv{a}_{\bar{n},\bar{\tau}}^{{\rm ext}}\right)^H \tilde{\mv{Y}}_{\bar{n},\bar{\tau}},
\end{align}where $\gamma_{\bar{n},\bar{\tau}}=\sqrt{p}/\left(pL+\frac{\rho}{\|\mv{x}_{\bar{n},\bar{\tau}}^{{\rm ext}}\|_2}\right)>0$. As a result, $(\mv{x}_{\bar{n},\bar{\tau}}^{{\rm ext}})^T$ should be linear to the vector $(\mv{a}_{\bar{n},\bar{\tau}}^{{\rm ext}})^H \tilde{\mv{Y}}_{\bar{n},\bar{\tau}}$. The remaining job is to find the value of $\gamma_{\bar{n},\bar{\tau}}$ such that (\ref{eqn:gradient}) is true. By substituting $\gamma_{\bar{n},\bar{\tau}}$ into (\ref{eqn:gradient}), it follows that
\begin{align}
& (\hspace{-2pt}-\hspace{-2pt}\sqrt{p}\hspace{-2pt}+\hspace{-2pt}pL)\gamma_{\bar{n},\bar{\tau}}\hspace{-2pt}\left(\mv{a}_{\bar{n},\bar{\tau}}^{{\rm ext}}\right)^H \hspace{-2pt}\tilde{\mv{Y}}_{\bar{n},\bar{\tau}}\hspace{-2pt} +\hspace{-2pt}\frac{\rho \gamma_{\bar{n},\bar{\tau}}\hspace{-2pt}\left(\mv{a}_{\bar{n},\bar{\tau}}^{{\rm ext}}\right)^H \hspace{-2pt} \tilde{\mv{Y}}_{\bar{n},\bar{\tau}}}{\left\|\gamma_{\bar{n},\bar{\tau}}\hspace{-2pt}\left(\mv{a}_{\bar{n},\bar{\tau}}^{{\rm ext}}\right)^H \hspace{-2pt} \tilde{\mv{Y}}_{\bar{n},\bar{\tau}}\right\|_2} \nonumber \\
= & (\hspace{-2pt}-\hspace{-2pt}\sqrt{p}\hspace{-2pt}+\hspace{-2pt}pL)\gamma_{\bar{n},\bar{\tau}}\hspace{-2pt}\left(\mv{a}_{\bar{n},\bar{\tau}}^{{\rm ext}}\right)^H\hspace{-2pt} \tilde{\mv{Y}}_{\bar{n},\bar{\tau}} \hspace{-2pt}+\hspace{-2pt}\frac{\rho \gamma_{\bar{n},\bar{\tau}}\hspace{-2pt}\left(\mv{a}_{\bar{n},\bar{\tau}}^{{\rm ext}}\right)^H \hspace{-2pt} \tilde{\mv{Y}}_{\bar{n},\bar{\tau}}}{\gamma_{\bar{n},\bar{\tau}}\hspace{-2pt}\left\|\left(\mv{a}_{\bar{n},\bar{\tau}}^{{\rm ext}}\right)^H \hspace{-2pt} \tilde{\mv{Y}}_{\bar{n},\bar{\tau}}\right\|_2}\hspace{-2pt} \label{eqn:1} \\
= & \hspace{-2pt}\left(\hspace{-2pt}-\hspace{-2pt}\sqrt{p}\hspace{-2pt}+\hspace{-2pt}pL\gamma_{\bar{n},\bar{\tau}}\hspace{-2pt}+\hspace{-2pt}\frac{\rho}{\left\|\left(\mv{a}_{\bar{n},\bar{\tau}}^{{\rm ext}}\right)^H \hspace{-2pt}\tilde{\mv{Y}}_{\bar{n},\bar{\tau}}\right\|_2}\right)\left(\mv{a}_{\bar{n},\bar{\tau}}^{{\rm ext}}\right)^H \tilde{\mv{Y}}_{\bar{n},\bar{\tau}} \label{eqn:2} \\ =&\mv{0},
\end{align}where (\ref{eqn:1}) is because $\gamma_{\bar{n},\bar{\tau}}>0$. Note that there exists a unique positive solution to (\ref{eqn:2}) if $\|(\mv{a}_{\bar{n},\bar{\tau}}^{{\rm ext}})^H \tilde{\mv{Y}}_{\bar{n},\bar{\tau}}\|_2>\frac{\rho}{\sqrt{p}}$, and this solution is denoted by (\ref{eqn:gamma}) in Theorem \ref{theorem1}. Therefore, if $\|(\mv{a}_{\bar{n},\bar{\tau}}^{{\rm ext}})^H \tilde{\mv{Y}}_{\bar{n},\bar{\tau}}\|_2>\frac{\rho}{\sqrt{p}}$, then the solution given in (\ref{eqn:proportional}) and (\ref{eqn:gamma}) minimizes the objective function of problem (\ref{eqn:problem 2}).

Next, consider the case when $\mv{x}_{\bar{n},\bar{\tau}}^{{\rm ext}}=\mv{0}$. In this case, according to (\ref{eqn:subgradient 2}), $\mv{0}$ minimizes problem (\ref{eqn:problem 2}) if
\begin{align}\label{eqn:optimal condition}
-\sqrt{p}\left(\mv{a}_{\bar{n},\bar{\tau}}^{{\rm ext}}\right)^H \tilde{\mv{Y}}_{\bar{n},\bar{\tau}}+\rho\left[\partial \left\|\mv{x}_{\bar{n},\bar{\tau}}^{{\rm ext}}\right\|_2\big|_{\mv{x}_{\bar{n},\bar{\tau}}^{{\rm ext}}=\mv{0}}\right]^T=\mv{0}.
\end{align}Note that $\partial \left\|\mv{x}_{\bar{n},\bar{\tau}}^{{\rm ext}}\right\|_2|_{\mv{x}_{\bar{n},\bar{\tau}}^{{\rm ext}}=\mv{0}}=\{\mv{g}:\|\mv{g}\|_2\leq 1\}$. As a result, (\ref{eqn:optimal condition}) is true only if $\|(\mv{a}_{\bar{n},\bar{\tau}}^{{\rm ext}})^H\tilde{\mv{Y}}_{\bar{n},\bar{\tau}}\|_2\leq \rho/\sqrt{p}$. In this case, $\mv{x}_{\bar{n},\bar{\tau}}^{{\rm ext}}=\mv{0}$ minimizes the objective function of problem (\ref{eqn:problem 2}).

To summarize, the optimal solution to problem (\ref{eqn:problem 2}) is given by (\ref{eqn:optimal solution}). Theorem \ref{theorem1} is thus proved.

\subsection*{B: Proof of Theorem 2}\label{appendix2}
Here, we provide a brief proof of Theorem \ref{theorem2}. First, the proof of the global optimality of Algorithm \ref{table2} is based on \cite[Theorem 2]{razaviyayn2013unified}. In particular, we need to check the following three conditions: 1) We set the upper bound function (for each block) in \cite{razaviyayn2013unified} to be the objective function in (\ref{eqn:problem 2}). Then, all conditions in Assumption 2 in \cite{razaviyayn2013unified} hold automatically. 2)
	It follows from Theorem \ref{theorem1} that the objective function in problem (\ref{eqn:problem 2}) is strongly convex and hence the solution to problem (\ref{eqn:problem 2}) is unique. 3) Because the nonsmooth term is decoupled among different blocks, the objective function in problem (\ref{eqn:problem 1}) is indeed regular. Combining 1), 2), and 3) together, it follows from \cite[Theorem 2]{razaviyayn2013unified} that every limit point of the iterates generated by Algorithm \ref{table2} is a stationary point of problem (\ref{eqn:problem 1}). Moreover, since problem (\ref{eqn:problem 1}) is convex, any stationary point is also a global solution. As a result, Algorithm \ref{table2} solves problem (\ref{eqn:problem 1}) globally. Moreover, the sublinear convergence rate shown in (\ref{eqn:sublinear}) directly follows from \cite[Theorem 2]{hong2017iteration}.

\subsection*{C: Proof of Theorem 3}\label{appendix3}
First, the complexity of solving problem (\ref{eqn:problem 2}) for any block $\mv{x}_{\bar{n},\bar{\tau}}^{{\rm ext}}$ based on (\ref{eqn:optimal solution}) is ${\cal O}((L+\tau_{\max})M)$. Since there are $(\tau_{\max}+1)N$ blocks, the total complexity of each iteration of Algorithm \ref{table2} to update all blocks once is ${\cal O}((L+\tau_{\max})\tau_{\max}MN)$. According to Theorem \ref{theorem2}, to get an $\epsilon$-optimal solution of problem (\ref{eqn:problem 1}), the total number of iterations should be in the order of $1/\epsilon$. Thus, the total complexity of Algorithm \ref{table2} to find an $\epsilon$-optimal solution is (\ref{eqn:complexity}).

\bibliographystyle{IEEEbib}
\bibliography{CIC}

\end{document}